\documentclass[a4paper,11pt,reqno]{amsart}
\usepackage{amsmath}
\usepackage{relsize}
\usepackage[noadjust]{cite}
\usepackage{color}
\usepackage[dvipsnames]{xcolor}
\usepackage{mathtools}
\usepackage[normalem]{ulem}
\usepackage{comment}
\newcommand\redout{\bgroup\markoverwith
{\textcolor{red}{\rule[.5ex]{2pt}{0.4pt}}}\ULon}

\usepackage{hyperref, enumitem}
\hypersetup{
  colorlinks   = true, 
  urlcolor     = blue, 
  linkcolor    = Purple, 
  citecolor   = red 
}
\usepackage{amsmath,amsthm,amssymb}

\usepackage[margin=2.5cm]{geometry}
\usepackage{caption} 
\usepackage{tikz-cd}
\usetikzlibrary{calc,fit,matrix,arrows,automata,positioning}
\usetikzlibrary{decorations.pathreplacing,angles,quotes}
\usepackage{stmaryrd}

\usetikzlibrary{fit}
\tikzset{%
  highlight/.style={rectangle,rounded corners,fill=red!15,draw,
    fill opacity=0.5,thick,inner sep=0pt}
}

\usepackage{array}
\newcolumntype{x}[1]{>{\centering\arraybackslash\hspace{0pt}}p{#1}}

\theoremstyle{definition}
\newtheorem{theorem}{Theorem}[section]
\newtheorem{definition}[theorem]{{{Definition}}}
\newtheorem{example}[theorem]{{{Example}}}

\newtheorem{remark}[theorem]{{{Remark}}}
\newtheorem{conjecture}[theorem]{{{Conjecture}}}
\newtheorem{question}[theorem]{{{Question}}}
\newtheorem{corollary}[theorem]{{{Corollary}}}
\newtheorem{proposition}[theorem]{{{Proposition}}}

\newtheorem{lemma}[theorem]{{{Lemma}}}

\newcommand{\C}{\mathcal C}

\newcommand{\F}{\mathbb F}
\newcommand{\K}{\mathbb K}
\newcommand{\R}{\mathbb R}
\newcommand{\LL}{\mathbb L}
\newcommand{\CC}{\mathbb{C}}

\newcommand{\Hh}{\mathbb{H}}
\newcommand{\OO}{\mathbb{O}}
\newcommand{\E}{\mathbb{E}}

\newcommand{\Fq}{\F_q}

\newcommand{\mV}{\mathcal V}
\newcommand{\mU}{\mathcal U}
\newcommand{\mB}{\mathcal B}
\newcommand{\mI}{\mathcal I}

\DeclareMathOperator{\GL}{GL}
\DeclareMathOperator{\supp}{supp}

\DeclareMathOperator{\rk}{rk}

\DeclareMathOperator{\dd}{d}

\DeclareMathOperator{\rs}{rowsp}
\DeclareMathOperator{\cs}{colsp}

\DeclareMathOperator{\wt}{wt}

\DeclareMathOperator{\mm}{m}
\DeclareMathOperator{\I}{I}
\DeclareMathOperator{\csup}{csupp}
\DeclareMathOperator{\rsup}{rsupp}
\DeclareMathOperator{\Gal}{Gal}
\DeclareMathOperator{\End}{End}

\newcommand{\drk}{\dd_{\rk}}
\newcommand{\st}{\,:\,}

\newcommand{\mat}{\K^{m \times n}}
\newcommand{\sqmat}{\K^{m \times m}}
\newcommand{\Fqk}{[m\times n, k]_{\K}}
\newcommand{\Fqkd}{[m\times n, k, d]_{\K}}

\newcommand{\Fmk}{[n, k]_{\LL/\K}}
\newcommand{\Fmkd}{[n, k, d]_{\LL/\K}}

\usepackage{scrextend}
\deffootnote{0em}{0em}{\thefootnotemark\quad}

\title{Rank-metric codes over arbitrary fields: Bounds and constructions}
\usepackage[foot]{amsaddr}

\author[A. Neri]{Alessandro Neri}
\address{Alessandro Neri, \textnormal{Department of Mathematics and Applications ``R. Caccioppoli'', University of Naples Federico II, Via Cintia 21, Monte S. Angelo, 80121 Napoli, Italy}}
 \email{alessandro.neri@unina.it}
\author[F. Zullo]{Ferdinando Zullo}
\address{Ferdinando Zullo, \textnormal{Department of Mathematics and Physics, University of Campania ``Luigi Vanvitelli'', Viale Lincoln 5, 81100 Caserta, Italy}}
 \email{ferdinando.zullo@unicampania.it}

\begin{document}

\maketitle

\begin{abstract}
  Rank-metric codes, defined as sets of matrices over a finite field with the rank distance, have gained significant attention due to their applications in network coding and connections to diverse mathematical areas. Initially studied by Delsarte in 1978 and later rediscovered by Gabidulin, these codes have become a central topic in coding theory. This paper surveys the development and mathematical foundations, in particular, regarding bounds and constructions of rank-metric codes, emphasizing their extension beyond finite fields to more general settings. We examine Singleton-like bounds on code parameters, demonstrating their sharpness in finite field cases and contrasting this with contexts where the bounds are not tight. Furthermore, we discuss constructions of Maximum Rank Distance (MRD) codes over fields with cyclic Galois extensions and the relationship between linear rank-metric codes with systems and evasive subspaces. The paper also reviews results for algebraically closed fields and real numbers, previously appearing in the context of topology and measure theory. We conclude by proposing future research directions, including conjectures on MRD code existence and the exploration of rank-metric codes over various field extensions.
\end{abstract}



\tableofcontents

\section*{Introduction}

Rank-metric codes are typically defined as sets of $m\times n$ matrices over a finite field, equipped with the rank distance, that is, the distance induced by the rank function, for which the rank distance between two matrices is defined as the rank of their difference.
Interest in these codes has increased since they were proposed as a useful tool in network coding.  This application has its origin in \cite{ahlswede2000network}, where it was shown that random network coding is a powerful tool for maximizing information flow in a noncoherent network with multiple sources and sinks. In this context, rank-metric codes are particularly well-suited for constructing efficient encoding schemes; see \cite{silva2008rank}.
Since then, rank-metric codes have become a vibrant area of research.
However, rank-metric codes were already introduced and studied much earlier. Delsarte first studied rank-metric codes in 1978 under the name \emph{Singleton systems}, in terms of bilinear forms \cite{de78}. A few years later, Gabidulin independently rediscovered them \cite{ga85a}.

The first applications appeared in Roth's 1991 paper \cite{roth1991maximum}, where rank-metric codes were used for correcting crisscross errors, and in the work of Gabidulin, Paramonov, and Tretjakov \cite{gabidulin1991ideals}, who proposed a cryptosystem based on these codes.

Building on these foundational works, rank-metric codes have revealed deep connections to various areas of mathematics and computer science, including finite geometry, tensors, skew algebras, semifields, matroids, and more (see, e.g., \cite{sheekey201913,polverino2020connections,byrne2019tensor,sheekey2020new,alfarano2023recursive,gluesing2025polynomial}). 

The aim of this survey is to demonstrate that the study of linear spaces of matrices over a (not necessarily finite) field -- hence, rank-metric codes over any field -- has profound connections with even more areas of mathematics. For instance, the study of rank-metric codes over the real numbers with minimum distance two appeared in the context of differential geometry, Clifford algebras and measure theory even before Delsarte's paper; see, e.g., \cite{adams1962vector, adams1965matrices, handel1970subspaces}.

In particular, in this paper we survey existing bounds on the parameters of a rank-metric code, analogous to the classical Singleton bound for Hamming-metric codes. Using a similar argument, it is possible to derive a general Singleton-like bound, and rank-metric codes that attain such a bound are called Maximum Rank Distance (MRD) codes. We {recall} constructions of MRD codes over fields that admit a cyclic Galois extension of degree $\max\{m,n\}$, demonstrating that the Singleton-like bound for these extensions is always achieved.  
We then focus on fields $\K$ admitting a degree $m$ (not necessarily Galois) field extension $\LL$. In such a case, one can represent $m\times n$ matrices over $\K$ as vectors of length $n$ over $\LL$, after a suitable choice of a $\K$-basis of $\LL$. It is then natural to consider the additional structure of $\LL$-linearity in this setting. With this linearity, one can obtain a more geometric point of view of $\LL$-linear rank-metric codes of dimension $k$ in $\LL^n$, exploring their correspondence with the so-called  systems, which are $\K$-subspaces of $\LL^k$ of $\K$-dimension $n$. This leads to the study of scattered and evasive subspaces in the broader context of general field extension, which do not necessarily involve finite fields.

At this point, we switch to study two cases of fields for which the Singleton-like bound is not tight in general, and for  which we can not always represent matrices as vectors over extension field: algebraically closed fields and the real numbers. First, we note that, in the case of algebraically closed fields, the Singleton-like bound is never tight. Instead, there exists another bound on the parameters of a rank-metric code that is proven to be tight in this context, and was given by Westwick \cite{westwick1972spaces}. The situation is markedly different when working with rank-metric codes over the real numbers. Indeed, for the purpose of studying vector fields over a sphere in differential geometry, Adams derived a completely different bound for the maximum dimension of a rank-metric code of $n\times n$ matrices with maximum rank distance $n$ in \cite{adams1962vector}, which involves the so-called Radon-Hurwitz numbers. However, in the general case, very little is known about rank-metric codes over the reals.

We conclude the paper by outlining three potential research directions for future study, which we consider as natural extensions of the results collected in this survey. First, we propose an MRD conjecture that aims to characterize the existence of MRD codes of $m\times n$ matrices over finite fields, which are linear over $\F_{q^m}$.  Second, we suggest exploring rank-metric codes $m\times n$ matrices over general fields admitting a field extension $\LL$ of degree $m$, and that are linear over $\LL$. Lastly, we recommend the study of rank-metric codes over the real numbers, for which no general bounds that are tight for every set of parameters are known.

\section{Rank-Metric Codes}

\subsection{Basic Definitions and Properties}
Let us consider the map
$$\begin{array}{rcl}\drk: \mat \times \mat & \longrightarrow & \mathbb N\\
(X,Y) & \longmapsto & \drk(X,Y):=\rk(X-Y).
\end{array}$$ 
It is easy to verify that $\drk$ defines a metric on $\mat$, which we call  \textbf{rank distance}.

\begin{definition}
 A \textbf{(matrix) rank-metric code} is a $\K$-linear subspace $\C \subseteq \mat$. If $\C \neq \{0\}$, then the \textbf{minimum rank distance} of $\C$ is the integer 
$$\drk(\C):= \min\{\rk(X) \st X \in \C, \ X \neq 0\} = \min\{\drk(X,Y) \st X,Y \in \C, \ X \neq Y\}.$$
\end{definition}

 From now on, we will refer to a rank-metric code $\C \subseteq \mat$ of dimension $k$  as an $\Fqk$ code. When also the minimum distance  $d$ is known, we will call it an $\Fqkd$ code.

Also in the rank metric, one can define the notion of equivalence of codes.

\begin{definition}\label{def:equiv_K_linear}
Two rank-metric codes $\C,\C' \subseteq \mat$ are \textbf{(linearly) equivalent} if there exists a $\K$-linear isometry $\varphi: (\mat,\mathrm{d_{rk}}) \to (\mat,\mathrm{d_{rk}})$ such that $\varphi(\C)=\C'$.
If $\C, \C'\in \mat$ are equivalent rank-metric codes, then we will write $\C\sim \C'$.
\end{definition}

As a linear isometry of $\mat$ is necessarily bijective, equivalent codes have the same dimension and minimum distance. According to \cite{hua1951theorem,wan1996geometry}, in which all the $\K$-(semi)linear isometries are classified,  we have the following result.
\begin{theorem}[see \textnormal{\cite{hua1951theorem,wan1996geometry}}]
 Let $\C,\C'\subseteq \mat$ be two rank-metric codes. Then, $\C\sim \C'$ if and only if there exist invertible matrices $A \in \GL(n,\K)$, $B\in \GL(m,\K)$ such that 
$$\C'=A \C B:=\left\{A X B \st X \in \C\right\},$$
or, when $m=n$,
$$\C'=A\C^\top B:=\left\{AX^\top B \st X \in \C\right\}.$$
\end{theorem}

Recall that the \textbf{trace-product} of $X,Y \in \mat$ is $\mathrm{Tr}(XY^\top)$. It is well-known and easy to see that the map
$(X,Y) \mapsto \mathrm{Tr}(XY^\top)$ defines a $\K$-bilinear, symmetric and nondegenerate form on $\mat$.

\begin{definition}\label{def:matrixdual}
The \textbf{dual} of an $\Fqk$ code is $$\C^\perp:=\{X \in \mat \st \mathrm{Tr}(XY^\top) =0 \mbox{ for all } Y \in \C\}.$$
Note that $\C^\perp$ is an $[m\times n, nm-k]_{\K}$ code.
\end{definition}

We also define the column support and row support of a rank-metric code.
See~\cite{gorla2021rank} for a detailed analysis of the various definitions of rank-support proposed in the literature.
For this purpose, it is important to introduce the notion of row space and column space of a matrix  $A \in \F^{m\times n}$, where $\F$ is a field. They are defined as the $\F$-subspace generated by the rows, respectively the columns, of $A$, and denoted by $\rs(A)$ and $\cs(A)$. 

\begin{definition}
	Let $\C$ be a rank-metric code. The \textbf{column support} and the \textbf{row support} of $\C$ are defined to be the $\K$-subspaces of $\K^m$ and $\K^n$, respectively,
	defined by 
	$$\csup(\C):=\sum_{M \in \C} \mathrm{colsp}(M), \qquad \rsup(\C):=\sum_{M\in \C} \mathrm{rowsp}(M),$$
	where the sums are sums of vector subspaces. The code $\C\subseteq \mat$ is said to be \textbf{column-nondegenerate} if $\csup(\C)=\K^m$ and \textbf{row-nondegenerate} if $\rsup(\C)=\K^n$.
\end{definition}

\subsection{Maximum Rank Distance Codes}

The following result is the rank-metric analogue of the Singleton bound for codes with the Hamming metric.
In order to show a short proof of this bound, we will use the following.

\begin{definition}
    Let $A \in \K^{m\times n}$ and $\mathcal{J}\subseteq [n]$, and let $A^{\mathcal{J}}$ denote the matrix in $\mathbb{K}^{m\times |\mathcal{J}|}$ obtained from $A$ by deleting the columns not indexed in $\mathcal{J}$. For a given $\mathcal I\subseteq [n]$, consider the map
    $$\begin{array}{rcl}
    \pi_{\mI}:\mat & \longrightarrow & \K^{m\times (n-|\mI|)} \\
    A & \longmapsto & A^{[n]\setminus \mI}.
\end{array}$$
    Let $\C$ be an $\Fqkd$ code, we define the \textbf{puncturing} of $\C$ on the coordinates in $\mathcal{I}$ as 
    \[ \pi (\C,\mathcal{I}):=\pi_{\mI}(\C)=\{ A^{[n]\setminus \mathcal{I}} \colon A \in \C \}. \]
\end{definition}

\begin{remark}
    The above notion of puncturing is not the most general one; for a more comprehensive definition, we refer the reader to \cite{byrne2017covering}. However, we will restrict to this specific definition throughout this paper, since we do not need the general notion of puncturing.
\end{remark}

\begin{theorem}[\textnormal{see  \cite{de78,ga85a}}]\label{thm:singbound}
    Let $\C$ be an $\Fqkd$ code. Then
    $$k\leq \max\{m,n\}(\min\{m,n\}-d+1).$$
\end{theorem}

\begin{proof}
Without loss of generality, let us assume that $n\leq m$. Let $\mI\subset[n]$ be any set with $|\mI|=d-1$.  
Consider the map $\pi_{\mI}$ restricted to $\C$:
$$\begin{array}{rcl}
    \pi_{\mI}:\C & \longrightarrow & \K^{m\times (n-d+1)} \\
    A & \longmapsto & A^{[n]\setminus \mI}.
\end{array}$$
The map $\pi_{\mI}$ is $\K$-linear, and by the assumption on the minimum rank distance on $\C$, it must be injective.
Thus,  
$$k=\dim_{\K}(\C)=\dim_{\K}(\pi(\C,\mI))\leq \dim_{\K}( \K^{m\times (n-d+1)})= m(n-d+1).$$
\end{proof}

\begin{definition}
A code $\C$ is  \textbf{maximum rank distance} (\textbf{MRD}) if it meets the bound of Theorem \ref{thm:singbound}, or if it is the zero code.
\end{definition}

\begin{remark}
The Singleton-like bound of  Theorem \ref{thm:singbound} is not always sharp. In particular, we will see in Section \ref{sec:complex_and_reals} that this depends on the underlying field $\K$.
\end{remark}

The property of being MRD satisfies a duality statement.

\begin{theorem}[\textnormal{see \cite{de78,ga85a}}]\label{thm:dualMRD}
Let $\C$ be an $\Fqk$  MRD code. Then $\C^\perp$ is an $[m \times n, nm-k]_{\K}$ MRD code.
\end{theorem}

We conclude by showing how one can construct MRD codes in $\mat$ by using MRD codes in $\K^{\max\{m,n\}\times \max\{m,n\}}$. 

\begin{proposition}\label{prop:puncturingMRD}
    Assume that $n\le m$ and $m-n<d$, and let $\C$ be an $[m\times m,m(m-d+1),d]_{\K}$ code. For every $\mathcal I\subseteq [n]$ with $|\mathcal I|=m-n$, the code 
    $\pi (\C,\mathcal{I})$
    is a $[m\times n, m(m-d+1),d-m+n]_{\K}$ MRD code.
\end{proposition}

As a consequence, if one can construct MRD codes in $\K^{m\times m}$ of minimum rank distance $d$ for every $d\le m$, then, using Proposition \ref{prop:puncturingMRD}, one obtains MRD codes in $\mat$ of minimum rank distance $d$ for every $d\le \min\{m,n\}$.
In next subsection we will construct MRD codes in $\K^{m\times m}$ for any possible value of the minimum distance under certain circumstances; Proposition  \ref{prop:puncturingMRD} will be used to construct non-square MRD codes.

\begin{remark}\label{rem:puncturing_MRD}
    Assume without loss of generality that $n\le m$. The reader may wonder whether every MRD code in $\K^{m \times n}$ of minimum rank distance $d$ can be obtained by puncturing an MRD code in $\K^{m\times m}$ of minimum rank distance $d+m-n$. The answer to this interesting question is not yet known, although the authors believe that this is not true in general. In Remark \ref{rem:2extscattered}, we will explore this question in the case of vector codes that have a stronger  linearity condition. We will see  surprisingly such a linearity can be used to construct examples of linear vector MRD codes that are not the puncturing of longer linear vector MRD codes.
\end{remark}


\subsection{Construction of MRD codes}\label{sec:constrMRD}

In this section we describe a construction of MRD codes in $\mat$ for every admissible minimum rank $d$, over fields admitting a cyclic Galois extension field of degree $\max\{m,n\}$. For simplicity, we will illustrate this construction in the square case, keeping in mind that it can be extended to the rectangular case by using Proposition \ref{prop:puncturingMRD}.

For the rest of this section, we will assume that the field $\K$ admits a degree $m$ field extension $\LL$ which is Galois, and whose Galois group is
$$G:=\Gal(\LL/\K).$$

Note that $\LL$ is a $\K$-vector space of dimension $m$, and thus, if we fix a $\K$-basis of $\LL$, we have a $\K$-algebra isomorphism
\begin{equation}\label{eq:isomorphism}
    \End_{\K}(\LL)\cong \sqmat.
\end{equation}
One of the key ideas for constructing MRD codes is thus considering the additional multiplicative structure of the auxiliary extension field $\LL$, by using the isomorphism \eqref{eq:isomorphism}.

The second crucial ingredient is the fact that such extension is Galois. This allows to express the $\K$-algebra $\End_{\K}(\LL)$ by using the skew algebra generated by $\LL$ and $G$. Formally, define
$$\LL[G]:=\left\{\sum_{g \in G}a_gg \,:\, a_g \in \LL\right\}.$$
The set $\LL[G]$ can be equipped with two operations: the sum, which is defined in the usual way, that is
$$\left(\sum_{g\in G}a_gg\right)+\left(\sum_{g\in G}b_gg\right)=\left(\sum_{g\in G}(a_g+b_g)g\right),$$
and the composition, which is defined over monomials as
$$(a_gg)\circ(b_hh)=a_gg(b_h)(g\cdot h), \qquad a_g,b_h \in \LL, g,h \in G,$$
-- where $\cdot$ denotes the group operation in $G$ -- and  it is then extended by distributivity. With these two operations $\LL[G]$ is a (noncommutative) ring and a $\K$-algebra.

The following result clarifies the importance of the skew algebra $\LL[G]$. It is a direct consequence of Artin's theorem of independence of characters \cite[Theorem 12]{artin2012galois}.

\begin{theorem}[{see \cite[Theorem 12]{artin2012galois}, \cite[Theorem 1.3]{ChHaRo/65}, \cite[Theorem 1]{gow2009galois}}]\label{thm:skew_isomorphism}
    Let $\LL/\K$ be a finite Galois extension with $G=\Gal(\LL/\K)$. Then, 
    $$\begin{array}{rccl}\psi:&\LL[G]& \longrightarrow & \End_{\K}(\LL) \\
    &\sum\limits_ga_gg & \longmapsto & \left\{\begin{array}{ccc}
    \LL & \longrightarrow & \LL \\   
        \beta & \longmapsto & \sum_ga_gg(\beta) 
    \end{array} \right.
    \end{array}$$
    is a $\K$-algebra isomorphism.
\end{theorem}
For the general theory of rank-metric codes over general Galois extension in this setting, we refer the reader to \cite{augot2021rank}. 

The last ingredient for constructing MRD codes is to add the additional property of having a Galois group which is cyclic. Indeed, when $G=\langle \sigma \rangle=\{\sigma^0=\mathrm{id},\sigma, \sigma^2,\ldots,\sigma^{m-1}\}$, every element  $a\in \LL[\sigma]:=\LL[\langle \sigma \rangle]$ can be expressed as
$$a:=\sum_{i=0}^{m-1}a_i\sigma^i.$$
With such a representation, we can define a natural notion of \textbf{$\sigma$-degree} for a nonzero element $a\in\LL[\sigma]$, given by
$$\deg_\sigma(a):=\max\{i\,:\, a_i \neq 0\}.$$
The notion of $\sigma$-degree is then crucial for the construction of MRD codes. This is explained by the following result, which relates the $\sigma$-degree of an element in $\LL[G]$ with the nullity of the corresponding $\K$-endomorphism of $\LL$ under the map $\psi$ of Theorem \ref{thm:skew_isomorphism}. This result was proven (in an equivalent formulation) in \cite{de78} for the case of finite fields, when $\sigma$ is the $q$-Frobenius element, and in \cite{guralnick1994invertible}, for general cyclic Galois extension. 
\begin{theorem}[{see  \cite{guralnick1994invertible}, \cite[Theorem 5]{gow2009galoisBis}}]\label{thm:sigmadegree_bound}
    Let $a\in \LL[\sigma]$ be nonzero. Then
    $$\dim_{\K}\ker(\psi(a)) \leq \deg_{\sigma}(a).$$
\end{theorem}
The reader may have noticed that Theorem \ref{thm:sigmadegree_bound} represents a noncommutative analogue of the schoolbook result that the number of roots of a nonzero polynomial over a field cannot exceed its degree. This is indeed the case, and for more general results in this direction we refer the interested reader to \cite{lam1988vandermonde} or \cite[Theorem 3.4]{neri2022twisted}.

At this point, we have all the ingredients needed to exhibit the most prominent construction of MRD codes. This was exhibited for finite fields first by Delsarte \cite{de78}  and later by Gabidulin \cite{ga85a}. However, the proofs of their results were strongly relying on the finiteness of the field. For general fields admitting a cyclic Galois extension, this construction was -- to the best of our knowledge -- first given by Guralnick \cite{guralnick1994invertible}, who used a different approach. Independently, it can also be found in the works of Roth \cite{roth1996tensor} and  Augot-Loidreau-Robert \cite{augot2018generalized}.

\begin{theorem}[{see \cite{de78,guralnick1994invertible}}]\label{thm:DG_construction}
    For every positive integer $d\le m$, the set
    $$\mathcal G_d:=\{ a \in \LL[\sigma] \,:\, \deg_\sigma(a) \le m-d \}\cup\{0\}\subseteq \LL[\sigma]$$
    defines an $[m\times m, m(m-d+1),d]_{\K}$ MRD code under the map $\psi$.
\end{theorem}

The MRD codes defined in Theorem \ref{thm:DG_construction} are known as \textbf{Delsarte-Gabidulin codes} (or just \textbf{Gabidulin codes}), due to the fact that this construction has been first independently exhibited by both Delsarte and Gabidulin. More recently, new constructions have been found, see e.g. \cite{lunardon2018generalized,sheekey2020new}. 

Combining Theorem \ref{thm:DG_construction} with Proposition \ref{prop:puncturingMRD}, we can state the following result, which concludes the question about the existence of MRD codes over certain fields.

\begin{theorem}\label{thm:existenceMRD}
    Let $n,m,d$ be positive integers such that $1\le d \le n \le m$, and let $\K$ be a field admitting a cyclic Galois extension of degree $m$. Then, there exists a $[m\times n, m(n-d+1),d]_{\K}$ MRD code. 
    In other words, over such a field, the Singleton-bound of Theorem \ref{thm:singbound} is tight for every admissible set of parameters.
\end{theorem}

\begin{proof}
    If $n=m$, then the code $\mathcal G_d$ of Theorem \ref{thm:DG_construction} is a $[m\times m,m(m-d+1),d]_{\K}$ MRD code. Thus, let us assume $n<m$ and let us call $e=m-n$ and $d'=d+e$. Again, using Theorem \ref{thm:DG_construction}, we can construct a $[m\times m,m(m-d'+1),d']_{\K}$ MRD code $\mathcal G_{d'}$. At this point, using Proposition \ref{prop:puncturingMRD} we obtain a
    $[m\times (m-e),m(m-d'+1),d'-e]_{\K}$ MRD code. Since $m-e=n$, $m-d'=n-d$ and $d'-e=d$, this is just an
    $[m\times n,m(n-d+1),d]_{\K}$ MRD code.
\end{proof}

\begin{remark}
It must be noted that Delsarte-Gabidulin codes given in Theorem \ref{thm:DG_construction} have a natural structure of $\LL$-linearity, since they are defined as $\LL$-subspaces of $\LL[\sigma]$. Such a linearity can be shown to be also preserved by the puncturing operation of Proposition \ref{prop:puncturingMRD}, and hence, all the codes which we have exhibited for establishing  Theorem \ref{thm:existenceMRD} possess a hidden $\LL$-linearity.  
\end{remark}

In the next section, we will deepen the concept of linearity of rank-metric codes over an extension field, and give additional points of view for codes with this additional property. 

\section{Linear MRD codes}\label{sec:linearMRD}

Theorem \ref{thm:DG_construction} exhibits a construction of MRD codes over a field $\K$ by means of an auxiliary extension field $\LL$. If one looks deeper at this construction, one can see that the codes $\mathcal G_d$  have an additional property of linearity, since they are not only $\K$-linear, but they also form an $\LL$-linear subspace of $\LL[\sigma]$. Motivated by this observation, in this section we will look at rank-metric codes that can be seen as linear over an extension field. In particular, we will look at them
in terms of vectors of $\mathbb{L}^n$, where $\mathbb{L}$ is an extension of $\mathbb{K}$ of degree $m$, as it was done by Gabidulin in \cite{ga85a} for finite fields.

\begin{definition}
For any $v=(v_1,\ldots,v_n)\in \mathbb{L}^n$, we define the \textbf{rank weight} as
\[ \wt(v)=\dim_{\mathbb{K}}(\langle v_1,\ldots,v_n \rangle_{\mathbb{K}}), \]
where $\langle v_1,\ldots,v_n \rangle_{\mathbb{K}}$ denotes the $\mathbb{K}$-span of the entries of the vector $v$.
The \textbf{rank distance} between two vectors $u,v \in \mathbb{L}^n$ is defined as $\dd(u,v)=\wt(u-v)$. 
\end{definition}

Now let $\Gamma=(\gamma_1,\ldots,\gamma_m)$ be an ordered basis of $\LL/\K$. Then, for any $v=(v_1,\ldots,v_n) \in \LL^n$, there exists $\Gamma(v)\in \K^{m\times n}$ such that
\[ v=(\gamma_1,\ldots,\gamma_m)\Gamma(v). \]

Notice that the $\K$-linear map defined by $\Gamma$, is a $\K$-linear isometry between the spaces $(\mathbb{K}^{m\times n},d_{\mathrm{rk}})$ and $(\mathbb{L}^{n},d)$. Hence we can equivalently define rank-metric codes in the space $(\mathbb{L}^{n},d)$.

\begin{definition}
Given a vector $v=(v_1,\ldots,v_n) \in \LL^n$, we define the \textbf{rank support} $\mathrm{supp}(v)$ of $v$ as $\mathrm{rowsp}(\Gamma(v))$.
This definition can be extended to a subspace of $\LL^n$ as follows.
Let $D$ be a subspace of $\LL^n$, we define the \textbf{rank support of $D$} as
\[ \mathrm{supp}(D)=\sum_{v \in D}\mathrm{supp}(v), \]
or equivalently $\mathrm{supp}(D) = \mathrm{rsupp}(\Gamma(D))$.
\end{definition}

We note that these definitions do not depend on the choice of the basis of $\LL/\K$ and the rank support of either a vector or of a vector space is a subspace of $\K^n$. Moreover, we clearly have that $\wt(v)=\dim_{\K}(\supp(v))$.

\begin{remark}
For any $\LL$-subspace $\C\subseteq \LL^n$ with $\C\ne \{0\}$ we have that $\Gamma(\C)$ is column-nondegenerate.
Indeed, let $v=(v_1,\ldots,v_n) \in\C\setminus\{0\}$. Without loss of generality, suppose that $v_1=1$ and consider the codewords $\beta_1 v,\ldots, \beta_m v$ of $\C$. It follows that the first column of $\Gamma(\beta_i v)$ will be the column with all zeros but a one in the $i$-th row.
Therefore,
\[ \mathrm{csupp}(\Gamma(\C))\supseteq \sum_{i=1}^m \mathrm{colspan}(\Gamma(\beta_i v))\supseteq \left\langle \begin{pmatrix} 1\\ 0 \\ \vdots \\ 0\end{pmatrix},\begin{pmatrix} 0\\ 1 \\ \vdots \\ 0\end{pmatrix},\ldots, \begin{pmatrix} 0\\ 0 \\ \vdots \\ 1\end{pmatrix} \right\rangle =\K^m. \]
Therefore, we will say that $\C$ is nondegenerate when $\Gamma(\C)$ is row-nondegenerate. Note that this corresponds to requiring that the $\K$-span of the columns of any generator matrix of $\C$ has dimension $n$. 
\end{remark}

\begin{definition}\label{def:L-linear_codes}
    An $[n,k,d]_{\LL/\K}$ (or $[n,k]_{\LL/\K}$) \textbf{rank-metric code} $\C$ is an $\mathbb{L}$-subspace of $\mathbb{L}^{n}$ equipped with the rank distance, where $d$ is its \textbf{minimum rank distance}, that is
    \[ d=\min\{ d(u,v) \, \colon \, u,v \in \C, u\ne v \}. \]
    Moreover, we define a \textbf{generator matrix} for $\C$ as a matrix in $\LL^{k\times n}$ whose rows are an $\LL$-basis of $\C$.
\end{definition}

In this section, we will restrict our attention to rank-metric codes that are $\mathbb{L}$-subspaces of $\mathbb{L}^{n}$ of dimension $k$ with $0<k<n$.

The map 
\[
\begin{array}{ccl}\LL^n \times \LL^n &  \longrightarrow & \LL \\ 
(u,v) &\longmapsto & uv^\top 
\end{array}\] 
is a symmetric and nondegenerate bilinear form on $\LL^n$.
The \textbf{dual} of an $[n,k,d]_{\LL/\K}$ code $\mathcal{C}$ is defined as 
\begin{equation}\label{eq:dualcodes}
\mathcal{C}^{\perp} =\left\{
v\in \LL^n \,\colon \, uv^\top=0 \mbox{ for all } v \in \C
\right\}.
\end{equation}

\begin{remark}
 We denoted by $\C^\perp$ the dual codes in both vector and matrix frameworks. This is not really an abuse of notation for the following reason. Let $\C$ be an $\Fmk$ code and let $\Gamma$ be a $\K$-basis of $\LL$. Consider the field trace of $\LL/\K$,
 $$\begin{array}{rcl}\mathrm{Tr}_{\LL/\K}:\LL& \longrightarrow &\K\\
 \alpha & \longmapsto & [\LL:\K(\alpha)]\sum\limits_{j=1}^{t}\sigma_j(\alpha),
 \end{array}$$
 where $t=\frac{m}{[\LL:\K(\alpha)]}$ and $\{\sigma_j(\alpha)\,:\, j\in [t]\}$
 is the set of roots of the minimal polynomial of $\alpha$ over $\K$.

 If we suppose that the extension $\LL/\K$ is separable (this is for instance the case of finite fields), then $\mathrm{Tr}_{\LL/\K}$ is nonzero and it induces a nondegenerate $\K$-bilinear form 
 $$\begin{array}{rcl}\phi: \LL\times \LL & \longrightarrow & \K \\
 (\alpha,\beta) & \longmapsto & \mathrm{Tr}_{\LL/\K}(\alpha\beta)\end{array}$$

Given a $\K$-basis $\Gamma=(\gamma_1,\ldots,\gamma_m)$ of $\LL$, let $\Gamma'=(\gamma_1',\ldots,\gamma_m')$ be the dual basis of $\gamma$ with respect to the $\K$-bilinear form $\phi$, that is, 
$$\phi(\gamma_i,\gamma_j')=\mathrm{Tr}_{\LL/\K}(\gamma_i\gamma'_j)=\begin{cases}
    1 & \mbox{ if } i=j, \\
    0 & \mbox{ if } i\neq j.
\end{cases}$$
Then, direct computations show that for every $u,v \in \LL^n$, it holds that
$$\mathrm{Tr}_{\LL/\K}(uv^\top)=\mathrm{Tr}(\Gamma(u)\Gamma'(v)^\top).$$
Thanks to this identity, if $v \in \C^\perp$, then $\Gamma'(v) \in \Gamma(\C)^\perp$, showing that $\Gamma'(\C^\perp)\subseteq \Gamma(\C)^\perp$. The opposite inclusion follows by a dimension argument. Thus, one has
$$\Gamma'(\C^\perp)= \Gamma(\C)^\perp.$$
In particular, $\Gamma(\mathcal{C}^\perp)$ and $\Gamma(\mathcal{C})^\perp$ are equivalent, and each one can be obtained from the other by multiplying by the change-of-basis matrix from $\Gamma$ to $\Gamma'$ or by its inverse.
For the finite field case, see \cite[Theorem 21]{ravagnani2016rank}.
\end{remark}

Similarly to what has been done for rank-metric codes in the matrix framework, we consider the following notion of equivalence of codes.

\begin{definition}\label{def:equiv_L_linear}
    Two $[n,k]_{\LL/\K}$ rank-metric codes $\C_1$ and $\C_2$ are \textbf{equivalent} if there exists an $\LL$-linear isometry $\varphi \colon (\LL^n,d) \longrightarrow (\LL^n,d)$ such that $\varphi(\C_1)=\C_2$.
\end{definition}

The $\LL$-linear isometries of the metric space $(\LL^n,d)$ have been classified; see \cite{berger2003isometries}. 

\begin{theorem}
    Two $[n,k]_{\LL/\K}$ rank-metric codes $\C_1$ and $\C_2$ are equivalent if and only if there exists $A \in \mathrm{GL}(n,\mathbb{K})$ such that
    \[\C_1=\C_2 A =\{vA \colon v \in \C_2\}.\]
\end{theorem}
\begin{remark}It must be noted that the definition of equivalence given in Definition \ref{def:equiv_L_linear} is stricter than the one given in Definition \ref{def:equiv_K_linear}. However, they are strongly related. Let us take $\C_1$ and $\C_2$ to be two $[n,k]_{\LL/\K}$ rank-metric codes.

First of all, if $\C_1$ and $\C_2$ are  equivalent according to Definition \ref{def:equiv_L_linear}, then there exists $B\in\GL(n,\K)$ such that $\C_1=\C_2B$. Hence, for every ordered basis $\Gamma$ of $\LL/\K$, we have $\Gamma(\C_1)=\Gamma(\C_2)B$, that is, $\Gamma(\C_1)$ and $\Gamma(\C_2)$ are equivalent according to Definition \ref{def:equiv_K_linear}.

On the other hand, assume that there exists an ordered basis $\Gamma=(\gamma_1,\ldots,\gamma_m)$ of $\LL/\K$ such that $\Gamma(\C_1)$ and $\Gamma(\C_2)$ are equivalent according to Definition \ref{def:equiv_K_linear}. Then, there exist $A\in \GL(m,\K), B\in \GL(n,\K)$ such that  $\Gamma(\C_1)=A\Gamma(\C_2)B$ or, when $n=m$, $\Gamma(\C_1)=A\Gamma(\C_2)^\top B$.  Assume that we are in the former case. In the special case that $n=m$ and that $\LL=\F_{q^n}$ and $\K=\F_q$ are finite fields, it was shown in \cite[Proposition 3.9]{sheekey2020rank} that $\C_1=\sigma^i(\C_2) B'$, where $B' \in \mathrm{GL}(n,\K)$, $i \in \{0,\ldots,n-1\}$ and the automorphism $\sigma:\LL \to \LL$ is the $q$-Frobenius automorphism. 
Therefore, over finite fields and when $n=m$, if $\Gamma(\C_1)$ and $\Gamma(\C_2)$ are equivalent according to Definition \ref{def:equiv_K_linear} (but not via the transpose), then there exists $i \in \{0,\ldots,n-1\}$ such that $\C_1$ and $\sigma^i(\C_2)$ are equivalent according to Definition \ref{def:equiv_L_linear}. 
\end{remark}

 {We now move on to the study of generalized weights in the rank metric, which are an important invariant of a rank-metric code originally introduced in \cite{oggier2012existence,ducoat2015generalized}, and have  applications in secure network coding \cite{kurihara2015relative}.

\begin{definition}
    Let $\C$ be a nondegenerate $[n,k]_{\LL/\K}$ rank-metric code.
    For every $r \in [k]$, the \textbf{$r$-th generalized rank weight} of $\C$ is the integer
    \[ d_r(\C):=\min \{\dim_{\K}(\mathrm{supp}(\mathcal D)) : \mathcal D\subseteq \C\,\, \text{with}\,\, \dim_{\LL}(D)=r\}. \]
\end{definition}

When we need to emphasize the generalized rank weights of $\mathcal{C}$, we will say that $\mathcal{C}$ is an $[n,k,(d_1,\ldots,d_k)]_{\LL/\K}$ code, where $d_i= d_i(\mathcal{C})$ for each $i \in [k]$. Note that $\mathcal{C}$ is nondegenerate if and only if $d_k(\mathcal{C})=n$.

We will need the following properties of the generalized weights, see e.g. \cite{kurihara2015relative,ducoat2015generalized}.

\begin{proposition}\label{prop:Wei}
    Let $\C$ be a nondegenerate $[n,k,d]_{\LL/\K}$  code. We have that
    \begin{enumerate}
        \item (Monotonicity) $d=d_1(\C)<d_2(\C)<\ldots<d_k(\C)=n$;
        \item (Wei-type duality) $\{ d_i(\C) \colon i \in [k]\} \cup \{ n+1-d_i(\C^\perp) \colon i \in [n-k] \}=[n]$.
    \end{enumerate}
\end{proposition}

\subsection{Systems and rank-metric codes}

The use of the geometric viewpoint for the classical Hamming-metric codes has provided useful insights in coding theory; see e.g. \cite{tsfasman2007algebraic}. In the rank metric, the geometric approach was  first developed in  \cite{randrianarisoa2020geometric}, by introducing the notion of $q$-systems.
Such a correspondence has been widely used to study for instance minimal rank-metric codes, one-weight codes, MRD codes and covering codes; see e.g. \cite{alfarano2022linear,randrianarisoa2020geometric,zini2021scattered,bonini2023saturating}.
In this section we are going to extend this notion to the framework of any finite extension of fields.

\begin{definition}
    An $[n,k,(d_1,\ldots,d_k)]_{\LL/\K}$ (or $[n,k]_{\LL/\K}$) \textbf{system} $\mathcal{U}$ is an $n$-dimensional $\mathbb{K}$-subspace of $\mathbb{L}^{k}$ such that $\langle \mathcal{U}\rangle_{\LL}=\LL^k$ and 
    \[ d_i=n-\max\{ \dim_{\K}(\mathcal{U}\cap \mathcal{H}) \colon \mathcal{H} \text{ is a } (k-i)\text{-dimensional } \LL\text{-subspace of } \LL^k \}, \]
    for any $i \in \{1,\ldots,k\}$.
    Moreover, we say that two systems $\mathcal{U}$ and $\mathcal{V}$ are \textbf{equivalent} if there exists an  $\phi$ in $\GL(k,\LL)$ such that $\phi(\mathcal{U})=\mathcal{V}$.
\end{definition}

We will show that there is a one-to-one correspondence between equivalence classes of nondegenerate $[n,k,(d_1,\ldots,d_k)]_{\LL/\K}$ rank-metric codes and equivalence classes of $[n,k,(d_1,\ldots,d_k)]_{\LL/\K}$ systems.

Let $\C$ be a nondegenerate $[n,k,(d_1,\ldots,d_k)]_{\LL/\K}$ rank-metric code and let $G\in \LL^{k\times n}$ be a generator matrix of $\C$.
Define $\mathcal{U}\subseteq \LL^k$ as the $\K$-columnspan of $G$. Since $\C$ is nondegenerate then $\dim_{\K}(\mathcal{U})=n$. 
Following \cite{randrianarisoa2020geometric} (see also \cite[Lemma 3.7]{alfarano2022linear} and \cite[Theorem 3.1]{neri2023geometry}), one can prove a relation between the supports of the subcodes of $\mathcal{C}$ and the intersection of the related system with $\LL$-subspaces of $\LL^k$.

For this
purpose, we need the following map
\[
\begin{array}{rccl}
\psi_{G}:& \mathbb{K}^n& \longrightarrow &\mathcal{U} \\
&\lambda & \longmapsto & \lambda G^\top,\end{array}\]
where $G$ is a generator matrix for $\C$.

\begin{lemma}\label{lem:supports}
    Let $\C$ be a nondegenerate $[n,k,(d_1,\ldots,d_k)]_{\LL/\K}$ rank-metric code and let $G\in \LL^{k\times n}$ be a generator matrix of $\C$. Let $\mathcal{U}\subseteq \LL^k$ be the $\K$-columnspan of $G$. For any $i$-dimensional subspace $\mathcal D=\langle v_1,\ldots,v_i \rangle_{\LL}$ of $\C$ we have that
    \[ \mathrm{supp}(\mathcal D)^\perp = \psi_G^{-1}(\mathcal{U}\cap \langle v_1, \ldots, v_i\rangle_{\LL}^\perp), \]
    and hence 
    \[ d_i=n-\max\{ \dim_{\K}(\mathcal{U}\cap \mathcal{H}) \colon \mathcal{H} \text{ is a } (k-i)\text{-dimensional } \LL\text{-subspace of } \LL^k \}, \]
    for any $i \in \{1,\ldots,k\}$.
\end{lemma}
\begin{proof}
    The case in which $\dim_{\LL}(\mathcal D)=1$ can be proved by following the proof of \cite[Theorem 3.1]{neri2023geometry}. 
    Now, suppose that $i>1$, we have that
    \[ \mathrm{supp}(\langle v_1, \ldots, v_i\rangle_{\LL})^\perp= \mathrm{supp}(v_1)^\perp\cap  \ldots\cap \mathrm{supp}(v_i)^\perp = \psi_G^{-1}(\mathcal{U}\cap \langle v_1\rangle_{\LL}^\perp) \cap \ldots \cap \psi_G^{-1}(\mathcal{U}\cap \langle v_i\rangle_{\LL}^\perp), \]
    from which we obtain the assertion using the case $i=1$. 
\end{proof}

From Lemma \ref{lem:supports}, we derive a relation between the generalized weights of $\C$ and the $d_i$'s of $\mathcal{U}$, proving that $\mathcal{U}$ is an $[n,k,(d_1,\ldots,d_k)]_{\LL/\K}$-system.

Conversely, if $\mathcal{U}$ is an $[n,k,(d_1,\ldots,d_k)]_{\LL/\K}$-system then we can consider the rank-metric code $\C$ having as a generator matrix the columns of a $\K$-basis of $\mathcal{U}$. Arguing as before, we can show that $\C$ is an $[n,k,(d_1,\ldots,d_k)]_{\LL/\K}$ rank-metric code. 
One can prove that this way of associating rank-metric codes to systems and systems to rank-metric codes give a one-to-one correspondence between equivalence classes of nondegenerate rank-metric codes and systems.

\begin{theorem}
    There is a one-to-one correspondence between equivalence classes of nondegenerate $[n,k,(d_1,\ldots,d_k)]_{\LL/\K}$ rank-metric codes and equivalence classes of $[n,k,(d_1,\ldots,d_k)]_{\LL/\K}$-systems.
\end{theorem}

\begin{remark}\label{rem:2extscattered}
Let us go to the question raised in Remark \ref{rem:puncturing_MRD} adding the assumption to work with $[n,k]_{\LL/\K}$ codes and their additional $\LL$-linearity. Is it true that every $[n,k,n-k+1]_{\LL/\K}$ MRD code can be obtained as the puncturing of a $[m,k,m-k+1]_{\LL/\K}$ MRD code? The answer to this question is negative. There exists a $[5,2,4]_{\F_{2^6}/\F_2}$ MRD code which is not the puncturing of a $[6,2,5]_{\F_{2^6}/\F_2}$ MRD code. This was observed in \cite{blokhuis2000scattered}; see also \cite[Example 3.2]{lavrauw2016scattered}. For $q$ odd power of $3$, a general construction of $[4,2,3]_{\F_{q^5}/\Fq}$ MRD codes which are not puncturing of a  $[5,2,4]_{\F_{q^5}/\Fq}$ MRD code is given in \cite{bartoli2026maximally}. All these results are obtained in the language of \emph{scattered subspaces}.
\end{remark}

\subsection{Evasive and scattered subspaces}

Scattered subspaces have been originally introduced in \cite{ball2000linear} and formalized in \cite{blokhuis2000scattered}. Subsequently this notion has been extended to the notion of scattered subspaces with respect to subspaces, see e.g. \cite{lunardon2017mrd,sheekey2020rank,csajbok2021generalising}, and later to the notion of evasive subspaces, see \cite{bartoli2021evasive,gruica2024generalised}.
Here we rephrase these definitions in the context of arbitrary field extensions.

\begin{definition}
    Let $h,r,k,n$ be nonnegative integers such that $h<k\leq n$.
    An $[n,k]_{\LL/\K}$-system $\mathcal{U}$ is said to be $(h,r)$-\textbf{evasive} if for any $h$-dimensional $\LL$-subspace $W$ of $\LL^k$ we have
    \[ \dim_{\K}(\mathcal{U}\cap W)\leq r. \]
    When $h=r$, we say that $\mathcal{U}$ is $h$-\textbf{scattered} (or scattered with respect to the $h$-dimensional subspaces) and when $h=r=1$ we say that $\mathcal{U}$ is \textbf{scattered}.
\end{definition}

The following property is easily verified.

\begin{proposition}\label{prop:oneless}
    If $\mathcal{U}$ is an $(h,r)$-evasive $[n,k]_{\LL/\K}$-system, then $\mathcal{U}$ is also $(h-1,r-1)$-evasive.
\end{proposition}

As done in \cite[Theorem 2.9]{bartoli2021evasive}, it is possible to show the following result on the direct sum of evasive systems. 

\begin{theorem}\label{th:directsum}
    If $\mathcal{U}_i$ is an $(h,r_i)$-evasive $[n_i,k_i]_{\LL/\K}$-system for $i \in \{1,2\}$, then $\mathcal{U}=\mathcal{U}_1\oplus\mathcal{U}_2$ is $(h,r_1+r_2-h)$-evasive $[n_1+n_2,k_1+k_2]_{\LL/\K}$-system. 
\end{theorem}

Let us see an example of an $h$-scattered system.

\begin{example}\label{ex:hscatt}
    Let $\LL/\K$ be a Galois extension of degree $m$ and let $\sigma$ be a generator of $\Gal(\LL/\K)$.
    Let $k$ be a positive integer and consider the following $\K$-subspace
    \begin{equation}\label{eq:exsigma} \mathcal{U}=\{ (\alpha,\sigma(\alpha),\ldots,\sigma^{k-1}(\alpha)) \colon \alpha \in \LL\}\subseteq \LL^k. \end{equation}
    It is easy to verify that $\mathcal{U}$ is a $(k-1)$-scattered $[m,k]_{\LL/\K}$-system.
    Let $h<k$ be a positive integer and suppose that $k=(h+1)t$, for some $t>0$. Consider 
    \[ \mathcal{U}=\{ (\alpha_1,\sigma(\alpha_1),\ldots,{\sigma^{h}(\alpha_1)},\ldots,\alpha_t,\sigma(\alpha_t),\ldots,{\sigma^{h}(\alpha_t)}) \colon \alpha_1,\ldots,\alpha_t \in \LL\}\subseteq \LL^k. \]
    It is easy to verify that $\mathcal{U}$ is an $[mt,k]_{\LL/\K}$-system. By using Theorem \ref{th:directsum}, since $\mathcal{U}$ is the direct sum of subspaces as in \eqref{eq:exsigma}, we have that $\mathcal{U}$ is $h$-scattered.
\end{example}

\subsection{Evasive and scattered subspaces and rank-metric codes}\label{sec:evasive-rankmetriccodes}

By making use of Proposition \ref{prop:oneless}, together with Proposition \ref{prop:Wei}, we can characterize the evasiveness property in terms of the generalized weights of an associated code to the system, as done in \cite{marino2023evasive}.

\begin{theorem}\label{thm:gen_weights_evasive}
    Let $\C$ be a nondegenerate $[n,k,(d_1,\ldots,d_k)]_{\LL/\K}$ rank-metric code and let $\mathcal{U}$ be any system associated with $\C$. 
    The following are equivalent:
    \begin{itemize}
        \item [(1)] $\mathcal{U}$ is $(h,r)$-evasive;
        \item [(2)] $d_{k-h}\geq n-r$;
        \item [(3)] $d_{r-h+1}(\C^\perp)\geq r+2$.
    \end{itemize}
    In particular, also the following are equivalent
    \begin{itemize}
        \item [(1)] $d_{k-h}= n-r$;
        \item [(2)] $\mathcal{U}$ is $(h,r)$-evasive but not $(h,r-1)$-evasive;
        \item [(3)] $d_{r-h+1}(\C^\perp)\geq r+2\geq d_{r-h}(\C^\perp)+2$.
    \end{itemize}
\end{theorem}

From Theorem \ref{thm:gen_weights_evasive}, one can deduce an upper bound on the dimension of an $h$-scattered system, as done in \cite{marino2023evasive}. Such an upper bound was originally established in \cite{blokhuis2000scattered,csajbok2021generalising} in the context of finite fields using a slightly different approach. 

\begin{corollary}\label{cor:hscattbound}
    Let $\mathcal{U}$ be an $h$-scattered $[n,k]_{\LL/\K}$-system with $k< n$. We have 
    \[n\leq \frac{km}{h+1}.\]
\end{corollary}

\begin{definition}
    If $\mathcal{U}$ is an $h$-scattered $[n,k]_{\LL/\K}$-system reaching the equality in Corollary \ref{cor:hscattbound}, then we say that $\mathcal{U}$ is a \textbf{maximum $h$-scattered} system.
\end{definition}

\begin{remark}
    The $h$-scattered systems in Example \ref{ex:hscatt} are actually maximum $h$-scattered systems.
\end{remark}

The next result shows how  $\LL$-linear MRD codes correspond to maximum $h$-scattered subspaces, and it was originally proved in \cite{zini2021scattered} using a different but equivalent formulation.

\begin{theorem}\label{thm:max_hscatt_iff_mrd}
    Let $\C$ be an $\Fmkd$ code and let $\mU$ be any of the $\Fmkd$ systems associated with $\C$. Assume that $k\ge 2$. $\C$ is MRD if and only if one of the following holds:
    \begin{itemize}
        \item[(1)] $\mU$ is $(k-1)$-scattered;
        \item[(2)] $\mU$ is maximum $h$-scattered for some $h \in \{1,\ldots,k-2\}$.
    \end{itemize}
\end{theorem}

\section{Rank-metric codes over algebraically closed fields and over the reals}\label{sec:complex_and_reals}

The Singleton-like bound of Theorem \ref{thm:singbound} is not always sharp. In this section we will see this fact, by analyzing two special cases: when $\K=\overline{\K}$ is algebraically closed and when $\K=\R$ is the field of real numbers. These two cases have been studied for mathematical purposes not related to coding theory. 

\subsection{Rank-metric codes over algebraically closed fields}

For algebraically closed fields the situation is actually worse than Theorem \ref{thm:singbound}. This is a consequence of \cite[Theorem 2.1]{westwick1972spaces}, where the variety of the space of matrices of bounded rank has been characterized. In the sequel, we state this improvement on the Singleton-like bound for algebraically closed fields, which was first observed by Roy Meshulam in a private communication to Ron Roth. The result can be found in \cite{roth1991maximum}. We include a proof for completeness.

\begin{theorem}[{see \cite[Section V]{roth1991maximum}}]\label{thm:algebraicallyClosed} 
 Let $\C$ be an $\Fqkd$ code. If $\K$ is algebraically closed, then 
  $$k \leq  (m-d+1)(n-d+1).$$
\end{theorem}

\begin{proof}
 Let $\mV_{m,n,r}$ be the set of all matrices in $\mat$ of rank at most $r$. $\mV_{n,m,r}$ is a \textbf{determinantal variety}, which is the zero locus of the ideal generated by all the $(r+1)\times (r+1)$ minors. It is known that $\mV_{m,n,r}$ is an algebraic variety of dimension $r(m+n-r)$ and positive degree. By assumption, we have that $\C \cap \mV_{m,n,d-1}=\{0\}$. Since $\K$ is algebraically closed, then this necessarily implies that $\dim(\C\cap \mV_{m,n,d-1})=0$. On the other hand, by the affine dimension theorem (see e.g. \cite[Proposition I.7.1]{hartshorne2013algebraic}),
 $$0=\dim(\C\cap \mV_{m,n,d-1})\geq \dim(\C)+\dim(\mV_{m,n,d-1})-mn=k+(d-1)(m+n-d+1)-mn,$$
 which implies the desired inequality.
\end{proof}

There is a known construction of rank-metric codes which makes use of MDS codes in the Hamming metric, whose parameters are quite good. The only condition is that the cardinality of $\K$ is at least $\max\{m,n\}-1$. Note that this is a sufficient condition for the existence of an MDS code in the Hamming metric, since in this case one can construct doubly extended Reed-Solomon codes.
Such a constructive result has been first shown in \cite{roth1991maximum} when $n=m$, and then in \cite{rees1996linear} in the general case. The case $m=n$ and $d=n-1$ was already proved in \cite{sylvester1986dimension}.

Assume that $n\leq m$. For $i \in \{-m+1,\ldots,n-1\}$ and $v_i\in \K^{t_i}$, where
$$t_i=\begin{cases}
m-i & \mbox{ if } -m+1\leq i \leq n-m-1,\\
n & \mbox{ if } n-m\leq i \leq 0, \\
n-i & \mbox{ if } 1\le  i \le n-1,
\end{cases}$$
denote by $D_i(v_i)$ the matrix with all zero entries except from the $i$-th diagonal, which contains the vector $v_i$. Here, the $0$th diagonal of a matrix $M=(m_{i,j})$ is the principal one, and, in general, the $i$th diagonal is $(m_{j,j+i}\,:\, j\in [n])$. Note that $D_i(v_i)$ is well-defined since, with these definitions, the length of the $i$th diagonal is precisely $t_i$. Furthermore, for a subspace $V_i \subseteq \K^{t_i}$, we write $D_i(V_i)$ to indicate the subspace of $\mat$ given by
$$  D_i(V_i)\coloneqq \{ D_i(v) \st v \in V_i \}.$$
Let $\mathcal I_d\coloneqq \{i \st t_i \geq d\}$. For every $i\in \mathcal I_d$, consider a $[t_i,t_i-d+1,d]_{\K}$ MDS code in the Hamming metric $\C_i$. Thus, define 
\begin{equation}\label{eq:construction_MDS_diagonals}\C:=\bigoplus_{i \in \mathcal I_d} D_i(\C_i).\end{equation}

\begin{theorem}[{see \cite{roth1991maximum,rees1996linear}}]\label{thm:construction_Rees}
The code $\C$ defined in \eqref{eq:construction_MDS_diagonals} is an $[n \times m, (m-d+1)(n-d+1), d]_{\K}$ code.
\end{theorem}

As a consequence of Theorem \ref{thm:construction_Rees}, we derive the existence of $[m\times n, (m-d+1)(n-d+1),d]_{\K}$ codes whenever we can construct MDS codes over the field $\K$ of length $t_i$, for each $i\in\mathcal I_d$.

\begin{corollary}\label{cor:generalLowerBound}
 For every $d \leq n \leq m$, and for every  field $\K$ such that $|\K| \geq n-1$, there exists an $[n \times m, (m-d+1)(n-d+1), d]_{\K}$ code. In particular, for algebraically closed fields, the bound of Theorem \ref{thm:algebraicallyClosed} is  tight for every choice of (admissible) $n,m,d$.
\end{corollary}

\subsection{Rank-metric codes over the reals}

Let us focus now on the case when $\K=\R$. This case is much more complicated, and very little is known in the literature. However, it has been widely studied due their numerous applications, for instance in algebraic topology \cite{hopf1940systeme,adams1962vector,james1957whitehead} and, in turn, in Clifford algebras and mathematical physics \cite{lounesto2001clifford}, and Young measures in PDE theory and calculus of variations \cite{bhattacharya1994restrictions}. We  start from the case of determining the largest dimension of a subspace of square matrices $\C\subseteq \R^{n\times n}$ such that each nonzero matrix is invertible. In other words, we wish to determine the largest integer $k$ such that there exists an $[n\times n, k, n]_{\R}$ code. This case has been completely solved, and found its motivation from the well-known problem of vector fields over spheres \cite{james1957whitehead,adams1962vector}. This is -- up to our knowledge -- the only special case that is completely solved. 

\begin{definition}
Let $n$ be a positive integer, and let us write $n=(2a+1)2^b$, for some nonnegative integers $a,b$. Furthermore, let $c,d$ be the unique nonnegative integers such that $b=c+4d$ and $0\leq c \leq 3$. The \textbf{$n$th Radon-Hurwitz number} \cite{radon1922lineare,hurwitz1923komposition} is the integer
$$\rho(n) \coloneqq 2^c+8d.$$
\end{definition}

 We start by showing that, for every $n \in \mathbb N$, there exists an $[n\times n, \rho(n),n]_{\R}$ code. This was shown by James in \cite{james1957whitehead}. In the following, we provide a revisited idea of the construction of such codes, by using the simpler language of \cite{adams1965matrices}.

\begin{lemma}\label{lem:aux_adams_1}
For every positive integer $n$, we have
$$\rho(16n)=\rho(n)+8.$$
\end{lemma}

We now recall the existence of three important division algebras over the reals. These are given by the complex numbers $\CC$, the quaternions $\Hh$ and the octonions $\OO$, which have $\R$-dimension $2$, $4$ and $8$, respectively.
Now, for $\mathbb E \in \{\R, \CC, \Hh, \OO\}$, consider the left multiplication 
$$ \begin{array}{ccc} \E \times \E^{n} & \longrightarrow & \E^{n} \\
(\alpha, v) & \longmapsto & \mm_{\alpha}(v)\coloneqq av.
\end{array} $$
One can easily see that each $\mm_{\alpha}$ is an $\R$-linear map which is invertible. 
\begin{lemma}\label{lem:aux_adams_2}
The left multiplicative actions of
\begin{enumerate}
    \item $\R$ on $\R^{n}$,
    \item $\CC$ on $\CC^{n}$
    \item $\Hh$ on $\Hh^{n}$
    \item $\OO$ on $\OO^{n}$
\end{enumerate}
produce, respectively,
\begin{enumerate}
    \item an $[n\times n,1,n]_{\R}$ code;
    \item an $[2n\times 2n,2,2n]_{\R}$ code;
    \item an $[4n\times 4n,4,4n]_{\R}$ code;
    \item an $[8n\times 8n,8,8n]_{\R}$ code.
\end{enumerate}
\end{lemma}

\begin{proof}
The left multiplication is $\R$-linear, and it is invertible, since these are all division algebras. The claim follows after checking the dimensions of $\E^{n}$ as $\R$-vector spaces.
\end{proof}

\begin{lemma}\label{lem:aux_adams_3}
Let $\C$ be an $[n\times n,k,n]_{\R}$ code. Then, there exists a $[16n\times 16n,k+8,16n]_{\R}$ code.
\end{lemma}

\begin{proof}
The proof is divided in two steps.  

First, for every $A \in \C$ and $\lambda  \in \R$, consider the matrix 
$$M(A,\lambda):=\begin{pmatrix} \lambda \I_n & A^\top \\A & -\lambda \I_n \end{pmatrix}\in \R^{2n \times 2n}.$$
It is easy to see that the space
$$\widetilde{\C}\coloneqq \{M(A,\lambda) \st A \in \C, \lambda \in \R\}$$
is a $(k+1)$-dimensional $\R$-subspace of $\R^{2n\times 2n}$ such that every nonzero matrix is invertible. Furthermore, every matrix in  $\widetilde{\C}$ is symmetric.

Now, let us consider the $\R$-vector space isomorphism $\OO\cong \R^8$ obtained by fixing an $\R$-basis $\mB$ of $\OO$. For an element $\alpha \in \OO$, let $N_{\alpha}\in \R^{8\times 8}$ be the matrix representing $\mm_\alpha$ with respect to the basis $\mB$. Let $T$ be the $7$-dimensional $\R$-subspace of $\OO$ consisting of all the purely imaginary elements. For each $B \in \widetilde{C}$ and $\alpha \in T$, define
$$ N(B,\alpha)\coloneqq B \otimes \I_{8}+\I_{2n}\otimes N_\alpha.$$
It is easy to see that the space
$$\widehat{\C}\coloneqq \big\{N(B,\alpha) \st B \in \widetilde{\C}, \alpha \in T\big\}$$
is a $(k+8)$-dimensional $\R$-subspace of $\R^{16n\times 16n}$ such that every nonzero matrix is invertible. Indeed, $N(B,\alpha)$ is the Kronecker sum of $N_\alpha$ and $B$, whose eigenvalues are the pairwise sums of the eigenvalues of $B$ and $N_{\alpha}$. Since $B$ is symmetric, all its eigenvalues are real, while $N_{\alpha}$ cannot have real eigenvalues, due to the choice of $\alpha$ being in $T$.
\end{proof}

Combining Lemmas \ref{lem:aux_adams_1}, \ref{lem:aux_adams_2} and \ref{lem:aux_adams_3}, we can prove the following.

\begin{theorem}[{see \cite{james1957whitehead,adams1965matrices}}]\label{thm:james:construction}
 For every positive integer $n$, there exists an $[n\times n, \rho(n), n]_{\R}$ code.
\end{theorem}

\begin{remark}[Vector fields on spheres]
    Here, we give a brief overview on the connection between $[n\times n,k,n]_{\R}$ codes and vector fields. Denote by $S^{n-1}=\{ v \in \mathbb{R}^n \colon \|v\|=1 \}$  the $(n-1)$-sphere in $\mathbb{R}^n$ centered in the origin. We recall that a \textbf{continuous vector field} on $S^{n-1}$ is a continuous map $\phi \colon S^{n-1}\rightarrow \mathbb{R}^n$ such that $v \cdot \phi(v)=0$ for every $v \in S^{n-1}$.
    Also, we say that the continuous vector fields on $S^{n-1}$ $\phi_1,\ldots,\phi_k$ are called \textbf{linearly independent} if $\{ \phi_1(v),\ldots,\phi_k(v) \}$ is a set of linearly independent vectors in $\mathbb{R}^n$ for every $v \in S^{n-1}$. The natural question that arises is the determination of the maximum number of linear independent continuous vector fields on $S^{n-1}$. A way to construct linear independent continuous vector fields on $S^{n-1}$ is the following. Let $\C$ be an $[n\times n,k,n]_{\R}$ code and let $\{A_1,\ldots,A_{k}\}$ be an $\R$-basis of $\C$. Consider the $[n\times n,k-1,n]_{\R}$ subcode $\C'=\langle A_1,\ldots,A_{k-1}\rangle_{\R}$.
    For any $i \in \{1,\ldots,k-1\}$, define $B_i=A_k^{-1}A_i$ and the vector field on $S^{n-1}$ as
    \[ \phi_i(v)=\mathrm{proj}_{v^\perp}(B_iv),  \]
    where $\mathrm{proj}_{v^\perp}$ denotes the orthogonal projection map onto $v^\perp$.
    Then, one can derive  that $\phi_1,\ldots,\phi_{k-1}$ are linearly independent vector fields on $S^{n-1}$. Indeed, assume that there exist $\lambda_1,\ldots,\lambda_{k-1}\in \R$ and $v \in S^{n-1}$ such that $$\lambda_1\phi_1(v)+\ldots+\lambda_{k-1}\phi_{k-1}(v)=0.$$
    Then 
    $$0=\mathrm{proj}_{v^\perp} ((\lambda_1B_1+\ldots +\lambda_{k-1}B_{k-1})v)=\mathrm{proj}_{v^\perp} (A_k^{-1}(\lambda_1A_1+\ldots +\lambda_{k-1}A_{k-1})v),$$
    which means that $v$ is an eigenvector for the matrix $A_k^{-1}A$, with $A=\lambda_1A_1+\ldots +\lambda_{k-1}A_{k-1}\in \C'$. Let $\lambda\in \R$ be the corresponding eigenvalue, then $\rk(A-\lambda A_k)<n$, which implies $\lambda=0$ and $A=0$, and, consequently, $\lambda_1=\ldots=\lambda_{k-1}=0$.
\end{remark}
    
    Clearly, by Theorem \ref{thm:james:construction}, we immediately have that the maximum number of linear independent continuous vector fields on $S^{n-1}$ is at least $\rho(n)-1$.
    Adams showed in \cite{adams1962vector} that this number cannot exceed $\rho(n)-1$ by making use of  homotopy theory and topological K-theory. His groundbreaking result automatically implies the following bound on the dimension of a $[n\times n,k,n]_{\R}$ code. 

\begin{theorem}[\textnormal{see  \cite[Theorem 1.1]{adams1962vector}}]\label{th:adams_bound}
 Let $\C$ be an $[n\times n,k,n]_{\R}$ code. Then
 $$ k \leq \rho(n).$$
\end{theorem}
    
}

What about more general values of $d<n$? Very little is known. One solved case is the following and concerns only some special values of $n,m,d$, so that a sort of Theorem \ref{thm:algebraicallyClosed} holds also over the reals. 

For three positive integers $a\leq b \leq c$, define 
\begin{equation}\label{eq:degree_det_var}
    \dd(b,c,a)\coloneqq \prod_{i=0}^{b-a-1}\frac{(c+i)!i!}{(a+i)!(c-a+i)!}.
\end{equation}

\begin{proposition}[see \cite{rees1996linear}]\label{prop:odd_degree}
Let $0 < d <n \leq m$ be three positive integers, and let $\C$ be an $[m\times n, k, d]_{\R}$ code. If $\dd(m,n,d-1)$ is odd, then 
$$k \leq (m-d+1)(n-d+1).$$
\end{proposition}

\begin{proof}
The space $\mV_{m,n,d-1}$ of matrices of rank at most $d-1$  is an algebraic variety of dimension $(d-1)(m+n-d+1)$ of degree $\dd(m,n,d-1)$. Since the degree of $\mV_{m,n,d-1}$ is odd, then it meets nontrivially every linear space of dimension strictly greater than $mn-(d-1)(m+n-d+1)=(m-d+1)(n-d+1)$. This concludes the proof. 
\end{proof}

\begin{remark}\label{rem:odd_degree}
The natural question that arises at this point is when $\dd(n,m,d-1)$ is odd. A condition for which this holds was found in \cite{rees1996linear}. More precisely, let $s$ be the unique integer such that 
$n-d+1\leq 2^s<2(n-d+1)$. Then, $$\dd(n,n,d-1)$$ is odd if and only if $n \equiv \pm (n-d+1) \pmod{2^{s+1}}$. 

We immediately verify that when $d=n$, the integer $d(n,n,n-1)$ is odd if and only if  $n \equiv \pm 1 \pmod 2$, and therefore, by Proposition \ref{prop:odd_degree}, when $n$ is odd any $[n\times n,k,n]_{\R}$ code must have dimension $k\leq 1$. This coincides with the upper bound of Theorem \ref{th:adams_bound}, since one has $\rho(n)=1$ if and only if $n$ is odd. 
\end{remark}

Another special case of rank-metric codes over the reals was highly studied due to the application in topology and measure theory. This is the case of square matrix codes ($n=m$) with minimum rank distance $d=2$; see \cite{handel1970subspaces,bhattacharya1994restrictions}. In this case of minimum rank distance $d=2$, one can state the following bound.

\begin{proposition}[see \cite{rees1996linear}]
 Let $\C$ be an $[m\times n, k, 2]_{\R}$ code. Then
 $$k\leq nm-1- \max\left\{ a+b \st a \leq n-1, b\leq m-1, \binom{a+b}{a} \mbox{ is odd } \right\}.$$
\end{proposition}

\section{Future research directions}

In this section, we explore potential research directions in the theory of MRD codes, highlighting open problems and conjectures that could shape future investigations. We begin by discussing a conjecture on the existence of MRD codes over finite fields, followed by an examination of the broader problem of their existence over arbitrary fields. Finally, we consider the intriguing case of rank metric codes over the real numbers.

\subsection{An MRD conjecture}

Consider the case of a finite field extension  $\F_{q^m}/\F_q$. By Theorem \ref{thm:existenceMRD}, we know that  $\F_q$-linear MRD codes in $\F_{q^m}^n$ exist for any possible value of the minimum distance. The existence of $[n,k]_{\F_{q^m}/\F_q}$ MRD codes, however, is less straightforward to determine, and it is still unknown for many parameter sets.
Using the notation of Theorem \ref{thm:max_hscatt_iff_mrd}, we know that $[n,k]_{\F_{q^m}/\F_q}$ MRD codes exist in the following cases.

\begin{theorem}\label{thm:existenceMRDlinear}
There exists an $[n,k]_{\F_{q^m}/\F_q}$ MRD code in the following cases:
\begin{itemize}
    \item[(1)] $n\leq m$, see \cite{de78,ga85a,roth1991maximum};
    \item[(2)] $n=mk/2$, see \cite{ball2000linear,blokhuis2000scattered,bartoli2018maximum,csajbok2017maximum};
    \item[(3)] $n$ divides $mk$ and $mk/n$ divides $k$, see \cite{csajbok2021generalising};
    \item[(4)] $h=m-3$ and $k$ is odd, see \cite{csajbok2021generalising};
    \item[(5)] $k\ne 5$, $m=6$ and $q$ is an odd power of $2$, see \cite{bartoli2024new}.
\end{itemize}
\end{theorem}

Theorem \ref{thm:existenceMRDlinear} is based on a collection of papers in which the existence of $\F_{q^m}$-linear MRD codes is proved in some cases. Quite a few of these instances are proved using the characterization result of MRD codes as maximum $h$-scattered spaces given in Theorem \ref{thm:max_hscatt_iff_mrd}. The existence of maximum $h$-scattered space seems plausible, due to all these partial results. Thus, we believe that it is not too rash to formulate the following conjecture. 

\begin{conjecture}[The linear MRD conjecture]
    For any $q,m,n,k$ positive integers such that $q$ is a prime power, $k\leq n$, and one of the following holds:
    \begin{itemize}
        \item[(1)] $n\leq m$;
        \item[(2)] $n>m$ and $n=\frac{mk}{m-d+1}$, for some integer $d$.
    \end{itemize}
    There exists an $[n,k]_{\F_{q^m}/\F_q}$ MRD code.
\end{conjecture}

As shown in Section \ref{sec:evasive-rankmetriccodes}, the geometric counterparts of MRD codes are scattered subspaces with respect to the hyperplanes and maximum $h$-scattered subspaces. More precisely, when $n\leq m$ the $q$-system associated with an $[n,k]_{\F_{q^m}/\F_q}$ MRD code is (not necessarily maximum) scattered with respect to the hyperplanes and when $n>m$ the system associated with an $[n,k,d]_{\F_{q^m}/\F_q}$ MRD code is a maximum $(m-d)$-scattered subspace of dimension $\frac{mk}{m-d+1}$. Solving the aforementioned conjecture would thus imply the existence of these scattered spaces, which has been studied even before their connection to rank-metric codes. These spaces are linked to a wide variety of objects, such as translation hyperovals, translation caps, linear sets, codes with few weights in the Hamming metric, and semifields; see, for example, \cite{lavrauw2016scattered}. Consequently, new instances of the conjecture could also have an impact on other areas.

\subsection{Rank-metric codes over any field}

The picture for general fields is still far from clear. As mentioned in the previous subsection, we have a relatively clear understanding for finite fields and a few other specific classes of fields. For instance, in the case of fields admitting a Galois extension, we described in Section \ref{sec:constrMRD} the family of Delsarte-Gabidulin codes, which provide MRD codes for all possible values of the minimum distance. These fields exhibit behavior very similar to that of finite fields. Indeed, the family of Delsarte-Gabidulin codes can be viewed as linear MRD codes, as described in Section \ref{sec:linearMRD}. 

These codes can be used to construct $[n, k]_{\LL/\K}$ MRD codes when $n \leq [\LL : \K]$, corresponding to the parameters of case (1) in the previous section, as well as $[n, k, d]_{\LL/\K}$ MRD codes with $n = [\LL : \K]k/([\LL : \K] - d + 1)$, corresponding to the parameters of case (3) of Theorem \ref{thm:existenceMRDlinear}. It would be interesting to explore whether the codes used to achieve the parameters in cases (2), (4), and (5) of Theorem \ref{thm:existenceMRDlinear} can also be formulated within the framework of Galois extensions. However, the proofs of the MRD property for these codes seem to rely heavily on the finite field context.
In particular, the rank arguments underlying the MRD property crucially depend on the linear independence criteria provided by Moore and Dickson matrices, whose behavior is tightly linked to cyclic Galois extensions. This phenomenon is much more complicated to exploit for more general Galois extensions, as shown in \cite{augot2021rank}.

Another possible approach to constructing additional examples of MRD codes is the use of scattered polynomials. For a recent survey on this topic, see \cite{longobardi2024scattered}. A similar definition and connection to MRD codes can be formulated in the context of Galois extensions. 

\begin{definition}
    A \textbf{scattered endomorphism} is an element $a \in \LL[G]$ such that the image of $\langle \mathrm{id}, a \rangle_{\LL}$  under the map $\psi$ of Theorem \ref{thm:skew_isomorphism} is an MRD code. 
\end{definition}
\begin{example}
If $\LL/\K$ is cyclic, with $\Gal(\LL/\K)=\langle \sigma \rangle$,
then $\sigma+\alpha\sigma^{m-1}$ is a scattered endomorphism  whenever $\mathrm{N}_{\LL/\K}(\alpha)\neq 1$. This can be proved using the same arguments as in \cite{sheekey2016new,sheekey2020new}; see also \cite{neri2022twisted}.
\end{example}

For other known scattered polynomials $\bar{a} \in \F_{q^m}[\sigma]$ in the finite field context, 
using the arguments in \cite[Section 2]{maazouz2021valued}, one should be able to find a generalization whenever $\LL/\K$ is an unramified extension of non-Archimedean local fields. 

\begin{proposition}\label{prop:unramifiedextentionlocal}
    Let $\LL/\K$ be an unramified extension of non-Archimedean local fields and let $\F_{q^m}/\F_q$ be the corresponding extension of the residue fields.
    Let $a \in \mathcal O_\LL[\sigma]$ -- where $\mathcal{O}_{\LL}$ denotes the ring of integers of $\LL$ -- be such that its reduction modulo the maximal ideal of $\mathcal O_\LL$ is equal to $\bar{a}$, and let $\C=\langle \mathrm{id},a\rangle_{\LL}$ and $\bar{\C}=\langle \mathrm{id}, \bar{a} \rangle_{\F_{q^m}}$.
    If $\psi(\bar{\C})$ is an MRD code in $\End_{\Fq}(\F_{q^m})$, then also $\psi(\C)$ is an MRD code in $\End_{\K}(\LL)$.
\end{proposition}
\begin{proof}
For any element $\lambda a + \mu \in \mathcal{C}$ such that $\lambda, \mu \in \mathcal{O}_{\LL}$ and at least one of them has valuation $0$, the reduction modulo the maximal ideal of $\mathcal{O}_{\LL}$ yields a nonzero element of the MRD code corresponding to the extension $\F_{q^m}/\F_q$. Thus, the minimum distance of $\bar{\C}$ is at most the minimum distance of $\mathcal{C}$, confirming that $\mathcal{C}$ is indeed an MRD code.
\end{proof}

It is natural to ask if there are other types of extensions for which a similar argument works, and whether the assumption on $\LL/\K$ being an unramified extension of non-Archimedean local fields can be relaxed. 

\begin{question}
    Can the hypothesis that $\LL/\K$ is an unramified extension of non-Archimedean local fields in Proposition \ref{prop:unramifiedextentionlocal} be replaced by a  weaker one?
\end{question}

\subsection{Rank-metric codes over the real numbers}
Rank-metric codes over the real numbers have been studied independently by Adams \textit{et al.} in \cite{adams1962vector,adams1965matrices} and James in \cite{james1957whitehead} in the case of square $n\times n$ matrices with minimum rank distance $n$. In this case, the existence of rank-metric codes of maximum possible dimension with respect to the bound of Theorem \ref{th:adams_bound} has been completely settled. However, the general case of $[m\times n,k,d]_{\R}$ codes is still widely open. The only cases in which an upper bound can be considered to be tight is when the value $\mathrm{d}(n,m,d-1)$ defined in \eqref{eq:degree_det_var} is odd. Indeed, in that case, Proposition \ref{prop:odd_degree} gives the same upper bound as the one for algebraically closed fields, that is,
$$k\le (n-d+1)(m-d+1),$$
and Theorem \ref{thm:construction_Rees} (or, equivalently, Corollary \ref{cor:generalLowerBound}) guarantees that there is a construction of $[m\times n,(n-d+1)(m-d+1),d]_{\R}$ codes meeting this bound with equality. A first natural question is the following. 

\begin{question}
Is it possible to characterize the triples $(m,n,d)$ such that $\mathrm{d}(m,n,d-1)$ is odd?
\end{question}

This has only been answered when $m=n$, as seen in Remark \ref{rem:odd_degree}, but the general case is unclear.

However, the main research direction concerns the completion of the theory of rank-metric codes over the real numbers, and, in particular, can be summarized by the following question.

\begin{question}
Can we provide a general upper bound
$$k\leq f(m,n,d)$$ 
for an $[m\times n,k,d]_{\R}$ code, that generalizes the upper bound of Theorem \ref{th:adams_bound} and that is tight for every choice of $m,n$ and $d$?
\end{question}

\section*{Acknowledgments}

We would like to thank the anonymous referees for their valuable comments, which  improved the readability of this survey.
The research was partially supported by the project ``COMPACT'' of the University of Campania ``Luigi Vanvitelli'' and by the Italian National Group for Algebraic and Geometric Structures and their Applications (GNSAGA - INdAM) and by the INdAM - GNSAGA Project \emph{Tensors over finite fields and their applications}, number E53C23001670001.
This research was also partially supported by Bando Galileo 2024 – G24-216.

\end{document}